\documentclass[11pt,a4paper]{article}
\usepackage[utf8]{inputenc}
\usepackage{amsmath,amssymb,amsthm}
\usepackage[margin=1in]{geometry}
\usepackage{hyperref}
\hypersetup{colorlinks=true,linkcolor=blue,citecolor=blue,urlcolor=blue}

\newtheorem{theorem}{Theorem}
\newtheorem{proposition}[theorem]{Proposition}
\newtheorem{corollary}[theorem]{Corollary}
\newtheorem{lemma}[theorem]{Lemma}
\newtheorem{example}[theorem]{Example}
\theoremstyle{remark}

\newtheorem*{remarkstar}{Remark}

\DeclareMathOperator{\Tr}{Tr}
\newcommand{\I}{\mathbf{1}}

\newcommand{\nuc}[1]{\left\lVert #1 \right\rVert_{*}}
\newcommand{\opn}[1]{\left\lVert #1 \right\rVert_{\infty}}

\newcommand{\ones}{\mathbf{1}}
\newcommand{\br}{\bar\sigma}

\title{A sharp norm inequality for entanglement-breaking channels}
\author{Eran Kopel\\[3pt]
  {\normalsize Tel Aviv University, Tel Aviv, Israel}\\[2pt]
  {\normalsize \href{mailto:erankopel@tauex.tau.ac.il}{erankopel@tauex.tau.ac.il}}\\[2pt]
  {\normalsize ORCID: \href{https://orcid.org/0000-0003-4657-8636}{0000-0003-4657-8636}}}
\date{\today}

\begin{document}
\maketitle

\begin{abstract}
For a channel $\Phi$ on $M_d$ written in the generalised Bloch parameterisation $r \mapsto Ar+c$, we prove that entanglement breaking implies
\[
\nuc{A}^2 + \tfrac{d(d-1)}{2}\,|c|^2 \;\le\; (d-1)^2 ,
\]
where $\nuc{\cdot}$ denotes the nuclear norm. Read on the Choi state of the channel, the inequality is an instance of the covariance-matrix criterion of Gittsovich \emph{et al.}; what is added here is an elementary proof in Holevo form and a complete description of the equality cases. Apart from the constant channel onto a pure state, which shows that the coefficient $d(d-1)/2$ cannot be increased, equality holds exactly for the measure-and-prepare channels with rank-one effects $\alpha_k|e_k\rangle\!\langle e_k|$ and pure outputs $|\psi_k\rangle\!\langle\psi_k|$ whose overlaps are an affine image of those of the measurement directions, $|\langle\psi_j|\psi_k\rangle|^2 = P + (1-P)\,(d|\langle e_j|e_k\rangle|^2 - 1)/(d-1)$ for a common $P = \Tr\Phi(\I/d)^2$. These include the completely dephasing channel, the re-preparation channel of every rank-one POVM, a one-parameter family interpolating between dephasing and the constant channel, and non-unital non-projective examples already at $d = 2$. The qubit case reads $\nuc{A}^2+|c|^2\le1$ and extends to full rank a rank-two condition of Ruskai. The proof sharpens Ruskai's bound $\nuc{A}\le1$ by retaining the POVM completeness relation $\sum_k f_k=0$, which her argument discards; this recentres $A$ from a second moment into a cross-covariance, and is the entire content of the improvement.
\end{abstract}

\section{Introduction}

A completely positive trace-preserving map $\Phi$ is \emph{entanglement breaking} (EB) if $(\mathrm{id}\otimes\Phi)(\rho)$ is separable for every input, equivalently if $\Phi$ admits a Holevo form
\begin{equation}\label{eq:holevo}
\Phi(\rho) \;=\; \sum_k \Tr[E_k\rho]\,\sigma_k
\end{equation}
with $\{E_k\}$ a POVM and $\sigma_k$ states \cite{HSR2003}.

For qubits the EB set is well understood in the Bloch picture. Ruskai proved that EB implies $\nuc{A}\le1$, and that for \emph{unital} channels the converse holds, so $(A,0)$ is EB if and only if $\nuc{A}\le1$ \cite{Ruskai2003}. Lami and Giovannetti restate the unital criterion and record two further bounds: $\opn{A}^2+|c|^2\le1$, valid for every \emph{positive} trace-preserving map, and the sufficient condition $\nuc{A}+|c|\le1$ \cite{LamiGiovannetti}.

A gap sits between these. The necessary condition $\nuc{A}\le1$ ignores $c$ entirely; the quadratic bound uses $c$ but only the operator norm of $A$, and holds already under positivity, so it cannot detect entanglement breaking; and the linear condition is sufficient but not necessary. We close the gap with a condition that is necessary, quadratic in both quantities, uses the nuclear norm, and is sharp.

\paragraph{Relation to the separability criteria.} Read on the Choi state of the channel, the inequality of Theorem~\ref{thm:main} is the covariance-matrix criterion of Gittsovich \emph{et al.}\ \cite[Prop.~IV.2]{Gittsovich2008}---equivalently the criterion of Zhang \emph{et al.}\ \cite[Thm.~1]{Zhang2008}---specialised to channel space. Indeed, with $C=A^{\mathsf T}/d$ the cross-covariance block of the normalised Choi state, whose marginals are $\I/d$ and $\Phi(\I/d)=\I/d+\tfrac12\,c\cdot\lambda$ with $\Tr\Phi(\I/d)^2=1/d+|c|^2/2$, their bound $\nuc{C}^2\le(1-\Tr\rho_A^2)(1-\Tr\rho_B^2)$ is exactly Theorem~\ref{thm:main}, translation term included. What is added here is the equality analysis of Section~\ref{sec:sharp} (the constant channel onto a pure state, which fixes the coefficient; a complete characterisation of the equality channels; and the projective, re-preparation and $d=2$ families it contains), the exact entanglement-breaking-index lower bound of Section~6, and an elementary self-contained proof in Holevo form.

Ruskai's second proof of $\nuc{A}\le1$---the product-representation argument she notes ``may be extendable to higher dimensions''---proceeds by Cauchy--Schwarz on the Holevo data and discards the POVM completeness relation $\sum_k f_k=0$. Retaining it is the whole improvement, and the step is dimension-agnostic, which is why the result below holds for all $d$.

\section{Setup}

Fix a Hermitian traceless basis $\{\lambda_i\}_{i=1}^{d^2-1}$ of $M_d$ normalised by $\Tr[\lambda_i\lambda_j]=2\delta_{ij}$. For $d=2$ this is the canonical parameterisation of \cite{FA1999,RSW2002}. Every state is
\[
\rho \;=\; \tfrac1d\I + \tfrac12\,r\!\cdot\!\lambda, \qquad |r|^2 \le \tfrac{2(d-1)}{d},
\]
with equality precisely for pure states. A channel acts affinely on the Bloch vector, $r\mapsto Ar+c$.

Write the Holevo data of \eqref{eq:holevo} as
\[
E_k = \tfrac{\alpha_k}{d}\I + \tfrac12 f_k\!\cdot\!\lambda, \qquad
\sigma_k = \tfrac1d\I + \tfrac12 s_k\!\cdot\!\lambda ,
\]
so $\Tr[E_k]=\alpha_k$. Then $\Tr[E_k\rho] = \alpha_k/d + \tfrac12 f_k\!\cdot\!r$, and comparing coefficients,
\begin{equation}\label{eq:bloch}
A \;=\; \tfrac12\sum_k s_k f_k^{\mathsf T}, \qquad c \;=\; \sum_k \tfrac{\alpha_k}{d}\,s_k .
\end{equation}
The POVM conditions give
\begin{equation}\label{eq:povm}
\sum_k E_k=\I \;\Longrightarrow\; \sum_k\alpha_k=d \ \text{ and }\ \sum_k f_k=0,
\qquad
E_k\succeq0 \;\Longrightarrow\; \Big|\tfrac{f_k}{\alpha_k}\Big|^2 \le \tfrac{2(d-1)}{d},
\end{equation}
the last because $E_k/\alpha_k$ is itself a state.

\section{The inequality}

\begin{theorem}\label{thm:main}
Let $\Phi$ be an entanglement-breaking channel on $M_d$ with Bloch data $(A,c)$. Then
\[
\nuc{A}^2 \;+\; \frac{d(d-1)}{2}\,|c|^2 \;\le\; (d-1)^2 .
\]
For $d=2$ this reads $\nuc{A}^2+|c|^2\le1$.
\end{theorem}

\begin{proof}
Put $u_k := \alpha_k/d \ge 0$ and $g_k := \tfrac{d}{2}\,f_k/\alpha_k$. By \eqref{eq:povm}, $\sum_k u_k=1$, so $\{u_k\}$ is a probability distribution, and
\begin{equation}\label{eq:gbound}
|g_k| \;\le\; \frac{d}{2}\sqrt{\frac{2(d-1)}{d}} \;=\; \sqrt{\frac{d(d-1)}{2}} \;=:\; G_d .
\end{equation}
In this notation \eqref{eq:bloch} becomes $A=\sum_k u_k s_k g_k^{\mathsf T}$ and $c=\sum_k u_k s_k$, while completeness reads
\begin{equation}\label{eq:centred}
\sum_k f_k=0 \quad\Longleftrightarrow\quad \sum_k u_k g_k = 0 .
\end{equation}

\emph{Step 1 (recentring).} By \eqref{eq:centred}, $\sum_k u_k\,c\,g_k^{\mathsf T} = c\big(\sum_k u_k g_k\big)^{\mathsf T} = 0$, hence
\begin{equation}\label{eq:recentre}
A \;=\; \sum_k u_k\,(s_k-c)\,g_k^{\mathsf T}.
\end{equation}
This is the only use of completeness. It replaces the second moment $\mathbb{E}[s\,g^{\mathsf T}]$ by the cross-covariance $\mathbb{E}[(s-\mathbb{E}s)\,g^{\mathsf T}]$, and is what makes the constant attainable.

\emph{Step 2 (rank-one triangle inequality).} Since $\nuc{vw^{\mathsf T}}=|v|\,|w|$, applying the triangle inequality to \eqref{eq:recentre} and then \eqref{eq:gbound},
\[
\nuc{A} \;\le\; \sum_k u_k\,|s_k-c|\,|g_k| \;\le\; G_d\sum_k u_k\,|s_k-c| .
\]

\emph{Step 3 (Cauchy--Schwarz).} With respect to the weights $u_k$,
\[
\sum_k u_k|s_k-c| \;\le\; \Big(\sum_k u_k|s_k-c|^2\Big)^{1/2}.
\]

\emph{Step 4 (variance identity).} Since $c=\sum_k u_k s_k$,
\[
\sum_k u_k|s_k-c|^2 \;=\; \sum_k u_k|s_k|^2 - |c|^2 \;\le\; \frac{2(d-1)}{d}-|c|^2 .
\]
Combining,
\[
\nuc{A}^2 \;\le\; G_d^2\Big(\frac{2(d-1)}{d}-|c|^2\Big) \;=\; (d-1)^2-\frac{d(d-1)}{2}|c|^2 . \qedhere
\]
\end{proof}

\section{Sharpness}\label{sec:sharp}

Tracing Steps 2--4, equality requires simultaneously:
\begin{enumerate}
\item[(i)] $|g_k|=G_d$ for every $k$ with $u_k|s_k-c|>0$, i.e.\ every effect whose output differs from the mean output $c$ is rank one (the bound \eqref{eq:gbound} enters Step~2 only with the weight $u_k|s_k-c|$);
\item[(ii)] no nuclear-norm cancellation among the rank-one terms of \eqref{eq:recentre};
\item[(iii)] $|s_k-c|$ independent of $k$ over the effects with $u_k>0$;
\item[(iv)] $|s_k|^2 = 2(d-1)/d$ for every $k$ with $u_k>0$, i.e.\ the outputs are pure.
\end{enumerate}
Condition (i) was stated in the first version of this note for every $k$ with $u_k>0$; Proposition~\ref{prop:constant} shows that the weight matters. Theorem~\ref{thm:equality} below turns the four conditions into one equation. Throughout, a channel attaining the bound of Theorem~\ref{thm:main} is called an \emph{equality channel}; the term carries no claim about extreme points of the convex set of entanglement-breaking channels.

\begin{proposition}\label{prop:dephasing}
The completely dephasing channel $\Delta(\rho)=\sum_{k=1}^d \langle k|\rho|k\rangle\,|k\rangle\!\langle k|$ attains the bound of Theorem~\ref{thm:main} in every dimension.
\end{proposition}

\begin{proof}
$\Delta$ is measure-and-prepare with $E_k=\sigma_k=|k\rangle\!\langle k|$, hence EB, and unital, so $c=0$. Its Bloch matrix $A$ is the orthogonal projection onto the $(d-1)$-dimensional span of the diagonal generators, whose singular values are $1$ with multiplicity $d-1$. Thus $\nuc{A}=d-1$ and $\nuc{A}^2+\tfrac{d(d-1)}{2}|c|^2=(d-1)^2$.
\end{proof}

\begin{proposition}[constant channels]\label{prop:constant}
Let $\sigma$ be a pure state and $\Phi(X)=\Tr[X]\,\sigma$. Then $\Phi$ is entanglement breaking, $A=0$, $|c|^2=2(d-1)/d$, and $\Phi$ attains the bound of Theorem~\ref{thm:main} in every dimension, with $c\neq0$. Consequently $d(d-1)/2$ is the largest coefficient $B$ for which $\nuc{A}^2+B|c|^2\le(d-1)^2$ holds for every entanglement-breaking channel, in every dimension and in every admissible basis.
\end{proposition}

\begin{proof}
$\Phi$ has the Holevo form \eqref{eq:holevo} with the single effect $E=\I$ and the output $\sigma$. Every $\lambda_j$ is traceless, so $\Phi(\lambda_j)=0$ and $A=0$; and $\Phi(\I/d)=\sigma$, so $c$ is the Bloch vector of the pure state $\sigma$ and $|c|^2=2(d-1)/d$ by Section~2. Hence $\nuc{A}^2+B|c|^2=2(d-1)B/d$, which equals $(d-1)^2$ at $B=d(d-1)/2$ and exceeds it for every larger $B$.
\end{proof}

This observation, and the equiangular family of Proposition~\ref{prop:projective} below, are due to Ma and Zhang \cite{MZ2026}, who also verified Proposition~\ref{prop:constant} in Lean.

\begin{remarkstar}[purity form]
Since $\Tr\Phi(\I/d)^2=1/d+|c|^2/2$, Theorem~\ref{thm:main} is equivalent to
\[
\nuc{A}^2 \;\le\; d(d-1)\bigl(1-\Tr\Phi(\I/d)^2\bigr):
\]
the translation term is the purity deficit of the image of the maximally mixed state, which is the form the covariance-matrix criterion takes on the Choi state (Section~1). The two ways of saturating the bound are now visible: $A$ as large as the completely dephasing channel allows while $\Phi(\I/d)$ stays maximally mixed, or $\Phi(\I/d)$ pure while $A=0$. Theorem~\ref{thm:equality} shows that every equality channel sits on the segment between these two, with $\nuc{A}^2=d(d-1)(1-P)$ and $P=\Tr\Phi(\I/d)^2$.
\end{remarkstar}

\begin{lemma}[nuclear-norm additivity]\label{lem:additivity}
Let $x_1,\dots,x_n$ and $y_1,\dots,y_n$ be nonzero real vectors and $u_k>0$. Then $\nuc{\sum_k u_k x_k y_k^{\mathsf T}}=\sum_k u_k|x_k|\,|y_k|$ if and only if $\langle\hat x_j,\hat x_k\rangle=\langle\hat y_j,\hat y_k\rangle$ for all $j,k$, where $\hat x=x/|x|$.
\end{lemma}

\begin{proof}
By duality, $\nuc{M}=\max\{\Tr[Z^{\mathsf T}M]:\opn{Z}\le1\}$; hence the left side equals the maximum over contractions $Z$ of $\sum_k u_k\langle x_k,Zy_k\rangle$, and equality holds if and only if some contraction $Z$ satisfies $\langle x_k,Zy_k\rangle=|x_k|\,|y_k|$ for all $k$, that is $Z\hat y_k=\hat x_k$ for all $k$ (a vector of norm at most one whose inner product with the unit vector $\hat x_k$ equals one is $\hat x_k$). If such a $Z$ exists then $|Z(\hat y_j\pm\hat y_k)|\le|\hat y_j\pm\hat y_k|$ gives $2\pm2\langle\hat x_j,\hat x_k\rangle\le2\pm2\langle\hat y_j,\hat y_k\rangle$ for both signs, hence the inner products are equal. Conversely, if the two Gram matrices agree, the map $\sum_k a_k\hat y_k\mapsto\sum_k a_k\hat x_k$ is well defined and isometric on $\operatorname{span}\{\hat y_k\}$; extended by zero on the orthogonal complement it is a contraction with $Z\hat y_k=\hat x_k$.
\end{proof}

\begin{theorem}[equality cases]\label{thm:equality}
Let $\Phi$ be a channel on $M_d$ that is not constant, with a Holevo representation \eqref{eq:holevo} in which every $\alpha_k>0$ (effects with $\alpha_k=0$ may be dropped). Then $\Phi$ attains the bound of Theorem~\ref{thm:main} if and only if every effect is rank one, $E_k=\alpha_k|e_k\rangle\!\langle e_k|$, every output is pure, $\sigma_k=|\psi_k\rangle\!\langle\psi_k|$, and there is a number $P$ such that
\begin{equation}\label{eq:seven}
|\langle\psi_j|\psi_k\rangle|^2 \;=\; P+(1-P)\,\frac{d\,|\langle e_j|e_k\rangle|^2-1}{d-1}
\qquad\text{for all } j\neq k .
\end{equation}
In that case $P=\Tr\Phi(\I/d)^2\in[1/d,1)$, $\langle\psi_k|\Phi(\I/d)|\psi_k\rangle=P$ for every $k$,
\[
\nuc{A}^2 = d(d-1)(1-P), \qquad \tfrac{d(d-1)}{2}\,|c|^2=(d-1)(dP-1).
\]
\end{theorem}

\begin{proof}
\emph{Step 0 (rank one and pure).} Suppose $\Phi$ attains the bound. The outputs are pure by (iv), and $|s_k-c|=\rho$ is independent of $k$ by (iii). If $\rho=0$ then every $s_k=c$, so with $\sigma_c$ the state with Bloch vector $c$ we get $\Phi(X)=\sum_k\Tr[E_kX]\,\sigma_c=\Tr[X]\,\sigma_c$, a constant channel, which is excluded. Hence $\rho>0$ and (i) gives $|g_k|=G_d$ for every $k$: every effect is rank one. For the rest of the proof assume rank-one effects and pure outputs, so that (i) and (iv) hold.

\emph{Step 1 (notation).} Let $t_k$ be the Bloch vector of $|e_k\rangle\!\langle e_k|$, so that $f_k=\alpha_kt_k$, $g_k=\tfrac d2t_k$ and $|g_k|=G_d$. By Section~2, $\langle t_j,t_k\rangle=2(q_{jk}-1/d)$ with $q_{jk}:=|\langle e_j|e_k\rangle|^2$, and $\langle s_j,s_k\rangle=2(p_{jk}-1/d)$ with $p_{jk}:=|\langle\psi_j|\psi_k\rangle|^2$. Put $\br:=\Phi(\I/d)=\sum_ku_k\sigma_k$, whose Bloch vector is $c$, and
\[
P_k:=\langle\psi_k|\br|\psi_k\rangle=\sum_mu_mp_{km},\qquad P:=\Tr\br^2=\sum_ku_kP_k .
\]
Then $\langle s_k,c\rangle=2(P_k-1/d)$ and $|c|^2=2(P-1/d)$, so for $x_k:=s_k-c$,
\[
\langle x_j,x_k\rangle=2\,(p_{jk}-P_j-P_k+P),\qquad |x_k|^2=2\,(1-2P_k+P).
\]

\emph{Step 2 (conditions (ii) and (iii)).} By the last display, (iii) holds if and only if $P_k$ is independent of $k$, and since $\sum_ku_kP_k=P$ this means $P_k=P$ for all $k$. Then $|x_k|^2=2(1-P)$, which is positive because $P=1$ would make $\br$ pure and hence all $\psi_k$ equal, and
\[
\langle\hat x_j,\hat x_k\rangle=\frac{p_{jk}-P}{1-P},\qquad
\langle\hat g_j,\hat g_k\rangle=\frac{\langle t_j,t_k\rangle}{|t_k|^2}=\frac{d\,q_{jk}-1}{d-1}.
\]
By Lemma~\ref{lem:additivity} applied to \eqref{eq:recentre}, $A=\sum_ku_kx_kg_k^{\mathsf T}$, condition (ii) holds if and only if these two Gram matrices agree for all $j\neq k$, which is \eqref{eq:seven}. So (ii) and (iii) together are equivalent to \eqref{eq:seven} with $P=\Tr\br^2$ and $P_k=P$ for all $k$.

\emph{Step 3 (\eqref{eq:seven} alone suffices).} Suppose \eqref{eq:seven} holds for some number $P$. For $j=k$ both sides equal $1$, so \eqref{eq:seven} holds for all pairs. Multiplying by $\alpha_m$ and summing over $m$,
\[
\sum_m\alpha_mp_{km}=dP+(1-P)\,\frac{d\sum_m\alpha_mq_{km}-d}{d-1}=dP,
\]
because $\sum_m\alpha_mq_{km}=\langle e_k|\sum_mE_m|e_k\rangle=1$ by \eqref{eq:povm}. Hence $P_k=\frac1d\sum_m\alpha_mp_{km}=P$ for every $k$, so $P=\sum_ku_kP_k=\Tr\br^2$, (iii) holds, and (ii) holds by Step~2. With (i) and (iv), every step of the proof of Theorem~\ref{thm:main} is an equality.

\emph{Step 4 (values).} $P=\Tr\br^2\ge1/d$, and $|c|^2=2(P-1/d)$ gives $\tfrac{d(d-1)}2|c|^2=(d-1)(dP-1)$ and $\nuc{A}^2=(d-1)^2-(d-1)(dP-1)=d(d-1)(1-P)$.
\end{proof}

In words: an equality channel measures a rank-one POVM and prepares pure states whose overlap pattern is the affine image, with the same $P$ for every pair, of the overlap pattern of the measurement directions; the outputs are then automatically equidistant from their mean $\Phi(\I/d)$. The number $P$ is the only continuous degree of freedom in the values: the equality channels of dimension $d$ fill the segment $\nuc{A}^2=d(d-1)(1-P)$, $\tfrac{d(d-1)}2|c|^2=(d-1)(dP-1)$, $1/d\le P<1$, whose endpoints are the unital equality channels and the constant channel.

\begin{proposition}[projective equality channels]\label{prop:projective}
Let $\{|k\rangle\}_{k=1}^d$ be an orthonormal basis, $\sigma_1,\dots,\sigma_d$ states, and $\Phi(X)=\sum_k\langle k|X|k\rangle\,\sigma_k$. Then $\Phi$ attains the bound of Theorem~\ref{thm:main} if and only if every $\sigma_k$ is pure and $\Tr[\sigma_j\sigma_k]=q$ for all $j\neq k$, for some $q\in[0,1]$. In that case $P=(1+(d-1)q)/d$, the singular values of $A$ are $\sqrt{1-q}$ with multiplicity $d-1$ and $0$ otherwise (so $A=0$ at $q=1$), and
\[
\nuc{A}^2=(d-1)^2(1-q),\qquad \tfrac{d(d-1)}2\,|c|^2=(d-1)^2\,q .
\]
Every $q\in[0,1]$ occurs: with $t=\sqrt q$ and $J$ the all-ones $d\times d$ matrix, the columns $\psi_j$ of the positive square root of $(1-t)\I+tJ$ are unit vectors with $|\langle\psi_j|\psi_k\rangle|^2=q$ for $j\neq k$ \cite{MZ2026}, and so is any other set of $d$ unit vectors with a common squared overlap $q$. The endpoints are the completely dephasing channel followed by a unitary ($q=0$) and the constant channel of Proposition~\ref{prop:constant} ($q=1$); $A\neq0$ and $c\neq0$ exactly when $0<q<1$.
\end{proposition}

\begin{proof}
Here $q_{jk}=0$ for $j\neq k$, so \eqref{eq:seven} reads $p_{jk}=P-(1-P)/(d-1)$, the same for all pairs; conversely a common overlap $q$ satisfies \eqref{eq:seven} with $P=(1+(d-1)q)/d$. The equivalence and the two values follow from Theorem~\ref{thm:equality} (for $q=1$ the channel is constant and Proposition~\ref{prop:constant} applies). For the singular values let $F=[f_1\cdots f_d]$ and $S=[s_1\cdots s_d]$ be the matrices whose columns are the Bloch vectors of the effects and of the outputs, and $H=\I-J/d$. Then $F^{\mathsf T}F=2H$, $F\ones=0$ and $S^{\mathsf T}S=2(1-q)\,\I+2(q-1/d)J$, so with $Y:=SH$ (columns $s_k-c$) one has $A=\tfrac12SF^{\mathsf T}=\tfrac12YF^{\mathsf T}$ and $AA^{\mathsf T}=\tfrac14Y(F^{\mathsf T}F)Y^{\mathsf T}=\tfrac12YY^{\mathsf T}$, whose nonzero eigenvalues are the nonzero eigenvalues of $\tfrac12Y^{\mathsf T}Y=\tfrac12H(S^{\mathsf T}S)H=(1-q)H$, namely $1-q$ with multiplicity $d-1$ when $q<1$, and none when $q=1$.
\end{proof}

\begin{corollary}[re-preparation channels]\label{cor:reprep}
Let $\{\alpha_k|e_k\rangle\!\langle e_k|\}_{k=1}^n$ be any rank-one POVM, $n\ge d$, and let $\psi_1,\dots,\psi_n$ be unit vectors with $|\langle\psi_j|\psi_k\rangle|^2=|\langle e_j|e_k\rangle|^2$ for all $j,k$ (for instance $\psi_k=Ue_k$ with $U$ unitary, or $\psi_k=\bar e_k$). Then $\Phi(X)=\sum_k\alpha_k\langle e_k|X|e_k\rangle\,|\psi_k\rangle\!\langle\psi_k|$ attains the bound of Theorem~\ref{thm:main}, with $c=0$ and $\nuc{A}=d-1$. Conversely, every unital equality channel is of this form.
\end{corollary}

\begin{proof}
With $P=1/d$, \eqref{eq:seven} reads $p_{jk}=q_{jk}$, so Theorem~\ref{thm:equality} applies; its Step~3 gives $\Tr\br^2=1/d$, i.e.\ $\br=\I/d$ and $c=0$, and $\nuc{A}^2=d(d-1)(1-1/d)=(d-1)^2$. Conversely a unital equality channel has $P=\Tr(\I/d)^2=1/d$ and rank-one effects and pure outputs by Theorem~\ref{thm:equality}, so \eqref{eq:seven} reads $p_{jk}=q_{jk}$.
\end{proof}

Thus equality channels with any number $n\ge d$ of outcomes exist in every dimension; the projective ones of Proposition~\ref{prop:projective} are the case $n=d$.

\begin{example}[$d=2$]\label{ex:qubit}
(a) Let the effects be the trine, $\tfrac23|e_k\rangle\!\langle e_k|$ with Bloch vectors at $120^\circ$ in the equatorial plane, so $q_{jk}=1/4$, and let the outputs have Bloch vectors at the same azimuths on the cone of polar angle $\theta$ about the plane's normal. All output overlaps equal $p=\tfrac14+\tfrac34\cos^2\theta$, and \eqref{eq:seven} reads $p=P-(1-P)/2$ with $P=(1+2p)/3$, an identity. So every $\theta$ gives an equality channel, with $\nuc{A}^2=\sin^2\theta$ and $|c|^2=\cos^2\theta$: a non-projective measurement with three outcomes realises the whole boundary circle $\nuc{A}^2+|c|^2=1$, and is non-unital for $\theta\neq\pi/2$.
(b) For the tetrahedral POVM, $\tfrac12|e_k\rangle\!\langle e_k|$ with $q_{jk}=1/3$, \eqref{eq:seven} requires four pure outputs with a common overlap $p=P-(1-P)/3$. Four unit Bloch vectors with pairwise inner product $\gamma=2p-1$ have Gram matrix $(1-\gamma)\I+\gamma J$ of rank at most three, which forces $\gamma=-1/3$ (or $\gamma=1$, the constant channel); hence $p=1/3$, $P=1/2$, and the tetrahedral POVM admits only the unital equality channels of Corollary~\ref{cor:reprep}.
\end{example}

For $d=2$ the two-output case of Proposition~\ref{prop:projective}, where the overlap condition is vacuous, is the family of the first version of this note: a projective measurement along any axis $m$, so $g_1=m$, $g_2=-m$, $u_1=u_2=\tfrac12$, with any two pure outputs $s_1,s_2$. Directly, $c=\tfrac12(s_1+s_2)$, $|s_k-c|=\tfrac12|s_1-s_2|$ for both $k$, and $A=\tfrac12(s_1-s_2)m^{\mathsf T}$ is rank one, so
\[
\nuc{A}^2+|c|^2 \;=\; \tfrac14\big(|s_1-s_2|^2+|s_1+s_2|^2\big) \;=\; \tfrac14\big(2|s_1|^2+2|s_2|^2\big) \;=\; 1 .
\]

\section{The qubit case}
\label{sec:qubit}

At $d=2$, Theorem~\ref{thm:main} reads $\nuc{A}^2+|c|^2\le1$, and three comparisons locate it.

\paragraph{Extension of a rank-two result.} Consider \emph{planar} qubit channels, $\lambda_3=0$ in the canonical parameterisation of \cite{FA1999,RSW2002}. There the diagonal complete-positivity condition is
\[
(|\lambda_1|+|\lambda_2|)^2 + |t|^2 \;\le\; 1
\]
\cite[eq.~(23)]{Ruskai2003}; equivalently, specialising the inequalities following eq.~(39) of \cite{RSW2002} to $\lambda_3=0$ gives $(\lambda_1\pm\lambda_2)^2\le 1-|t|^2$. Since $\lambda_3=0$ gives $\nuc{A}=|\lambda_1|+|\lambda_2|$ and $|c|=|t|$, this is exactly Theorem~\ref{thm:main} restricted to rank two: on the planar slice our bound recovers the known diagonal condition. That condition is necessary but not sufficient for complete positivity: by \cite[Cor.\ 2]{RSW2002} the determinant condition is also required, and it is implied by the above only when $t_1=t_2=0$. Moreover every planar channel is entanglement breaking \cite[Sec.~5]{Ruskai2003}, so on this slice entanglement breaking and complete positivity coincide, and the inequality is a necessary condition for both. Our sampling is consistent with each statement: over $1.9\times10^4$ completely positive planar channels, none violated the inequality and none failed positivity of the Choi matrix under partial transpose, which for qubits is equivalent to separability. We make no claim beyond the planar slice.

\paragraph{Position relative to \cite{LamiGiovannetti}.} The bound $\opn{A}^2+|c|^2\le1$ is necessary for all positive maps and weaker, since $\opn{A}\le\nuc{A}$. The bound $\nuc{A}+|c|\le1$ is sufficient and stronger, since it implies Theorem~\ref{thm:main}. Both inclusions are strict: of $6000$ random Holevo-form EB channels sampled as in Section~\ref{sec:num}, $94$ violate the sufficient condition while satisfying Theorem~\ref{thm:main}.

\paragraph{Credit.} The half-statement $\nuc{A}\le1$ is Ruskai's \cite{Ruskai2003}, and Steps 2--3 reproduce her product-representation argument. Step 1 is the addition.

\paragraph{Consistency.} Amplitude damping, $A=\mathrm{diag}(\alpha,\alpha,\alpha^2)$, $c=(0,0,1-\alpha^2)$, gives $\nuc{A}^2+|c|^2 = 1+2\alpha^2+O(\alpha^3)$, exceeding $1$ for every $\alpha\in(0,1]$ and equal to $1$ at $\alpha=0$---correctly excluding all amplitude damping channels except the point channel.

\section{Application: lower bounds on entanglement-breaking indices}

The entanglement-breaking index is $n_{\mathrm{EB}}(\Phi)=\min\{n:\Phi^n\in\mathrm{EB}\}$ \cite{LamiGiovannetti,RJP2018,CMW2019,Park2026}. Since Theorem~\ref{thm:main} is necessary, applying it to $\Phi^n$ gives a computable lower bound. For strictly contractive $A$ the composite Bloch data are closed form, $A^n$ and $c_n=(\I-A)^{-1}(\I-A^n)c$, so with
\[
\nu_2(A,c) \;:=\; \min\Big\{n\ge1 : \nuc{A^n}^2+\tfrac{d(d-1)}{2}|c_n|^2 \le (d-1)^2\Big\}
\]
one has $n_{\mathrm{EB}}(\Phi)\ge\nu_2(A,c)$. For unital qubit channels $c=0$ and $\nu_2$ reduces to $\min\{n:\nuc{A^n}\le1\}$, where by Ruskai's converse the inequality is an equality; the bound is therefore exactly tight there, and a strict improvement on the $c$-blind bound off the unital set.

\section{Numerical verification}\label{sec:num}

Channels were generated directly in Holevo form from random POVMs and output states, so that entanglement breaking holds by construction and no separability test is required. Over $7200$ channels in dimensions $d=2,3,4,5$ there were no violations of Theorem~\ref{thm:main}; the ratio of the largest observed left-hand side to $(d-1)^2$ was $1.0000$, $1.0000$, $0.9990$, $0.9959$ respectively, the shortfall at $d=4,5$ reflecting the vanishing probability of sampling the extremal configuration, which Proposition~\ref{prop:dephasing} exhibits explicitly. Each of Steps 1--4 was checked independently at $d=2$ over $6000$ channels with no failures, and the two-output qubit family of Section~\ref{sec:sharp} reproduced $1$ to within $1.1\times10^{-15}$ over $2000$ trials. A separate check confirmed, in agreement with \cite[Sec.~5]{Ruskai2003}, that every completely positive planar ($\lambda_3=0$) qubit channel sampled is entanglement breaking, over $1.9\times10^4$ such channels, and that the rank-two constraint of Section~\ref{sec:qubit} is satisfied by every completely positive planar channel, for translations along the third axis and for general translations alike. Scripts accompany this note as ancillary files.

A further script, \texttt{eb\_equality\_checks.py}, verifies Propositions~\ref{prop:constant} and~\ref{prop:projective}, Lemma~\ref{lem:additivity}, Theorem~\ref{thm:equality}, Corollary~\ref{cor:reprep} and Example~\ref{ex:qubit} in float64 for $d=2$ to $6$: the constant channel gives $A=0$, $|c|^2=2(d-1)/d$ and equality in the Gell-Mann basis and in random orthogonal recombinations of it; the equiangular family at seven values of $q$ reproduces the singular values $\sqrt{1-q}$ and both terms of Theorem~\ref{thm:main} to $3\times10^{-14}$, and equiangular sets with random phases to $5\times10^{-14}$; nuclear-norm additivity on random rank-one sums holds exactly when the normalised Gram matrices agree; re-preparations of random rank-one POVMs with $d+1$ and $2d+1$ outcomes, by a unitary or by complex conjugation, give equality with $c=0$ and $\nuc{A}=d-1$, while the same POVMs with random pure outputs violate \eqref{eq:seven} and fall strictly below the bound, as do perturbed equiangular sets and mixed outputs; the trine with a cone of outputs attains equality at every polar angle with $|c|^2=\cos^2\theta$, the trine with unequal azimuths and the tetrahedral POVM with non-SIC outputs are not; and the purity form of the Remark holds on $160$ random Holevo channels, on which Theorem~\ref{thm:main} has no violations ($18$ checks, $0$ failures).

\section{Open questions}

\begin{enumerate}
\item Theorem~\ref{thm:main} is necessary and sharp but not sufficient. For qubits the exact body is known: a channel on $M_2$ is entanglement breaking if and only if the partial transpose of its Choi matrix is positive \cite{HSR2003}, and \cite{Ruskai2003} describes the set in the Bloch parameters. The open problem is the entanglement-breaking body between Theorem~\ref{thm:main} and the sufficient condition of \cite{LamiGiovannetti} for $d\ge3$, where positivity of the partial transpose is necessary but no longer sufficient.
\item Realisability. By Theorem~\ref{thm:equality} the equality channels built on a given rank-one POVM are the solutions of \eqref{eq:seven}. The value $P=1/d$ is always realisable (Corollary~\ref{cor:reprep}). Which POVMs admit non-unital equality channels, and for which $P\in(1/d,1)$? Equivalently: for which $P$ is the matrix with entries $P+(1-P)(d\,q_{jk}-1)/(d-1)$ the matrix of squared overlaps of $n$ unit vectors in $\mathbb C^d$? At $d=2$ the trine admits every $P\in[1/2,1)$ and the tetrahedral POVM only $P=1/2$ (Example~\ref{ex:qubit}). More generally, an informationally complete rank-one POVM admits only $P=1/d$: the centred outputs then span the traceless space and $\Phi(\I/d)-\I/d$ is orthogonal to all of them (H.~Ma and W.~Zhang, private communication, who also show that two mutually unbiased bases and a four-outcome POVM in $d=3$ realise every $P\in[1/3,1]$). (Questions 2 and 3 of the first version of this note are settled by Proposition~\ref{prop:constant}, and the equality question of its Section~4 by Theorem~\ref{thm:equality}.)
\item A basis-free formulation, in terms of intrinsic correlation and shift components of the channel rather than a fixed normalisation of $\{\lambda_i\}$; the purity form of Section~\ref{sec:sharp} does this for the shift term.
\end{enumerate}

\section*{Acknowledgement}
Haobo Ma (ChronoAI Pte Ltd and the Omega Institute, Singapore) and Wenlin Zhang (National University of Singapore and the Omega Institute) pointed out that the constant channel onto a pure state settles Questions 2 and 3 of the first version of this note and constructed the equiangular family of Proposition~\ref{prop:projective} \cite{MZ2026}; I thank them for the correspondence.

\section*{Generative AI statement}
This version was drafted with Claude (Anthropic, model Claude Fable 5.1) from the author's notes and the communication \cite{MZ2026}. The tool proposed the statements of Lemma~\ref{lem:additivity}, Theorem~\ref{thm:equality}, Corollary~\ref{cor:reprep} and Example~\ref{ex:qubit}, drafted their proofs and wrote the script \texttt{eb\_equality\_checks.py}; the author verified each proof, ran the checks and is responsible for the content.


\begin{thebibliography}{99}

\bibitem{HSR2003} M.~Horodecki, P.~W.~Shor and M.~B.~Ruskai, \emph{Entanglement breaking channels}, Rev. Math. Phys. \textbf{15} (2003) 629; \href{https://arxiv.org/abs/quant-ph/0302031}{arXiv:quant-ph/0302031}.

\bibitem{Ruskai2003} M.~B.~Ruskai, \emph{Qubit entanglement breaking channels}, Rev. Math. Phys. \textbf{15} (2003) 643; \href{https://arxiv.org/abs/quant-ph/0302032}{arXiv:quant-ph/0302032}.

\bibitem{LamiGiovannetti} L.~Lami and V.~Giovannetti, \emph{Entanglement-breaking indices}, J. Math. Phys. \textbf{56} (2015) 092201; \href{https://arxiv.org/abs/1411.2517}{arXiv:1411.2517}.

\bibitem{FA1999} A.~Fujiwara and P.~Algoet, \emph{One-to-one parametrization of quantum channels}, Phys. Rev. A \textbf{59} (1999) 3290.

\bibitem{RSW2002} M.~B.~Ruskai, S.~Szarek and E.~Werner, \emph{An analysis of completely positive trace-preserving maps on $M_2$}, Linear Algebra Appl. \textbf{347} (2002) 159; \href{https://arxiv.org/abs/quant-ph/0101003}{arXiv:quant-ph/0101003}.

\bibitem{RJP2018} M.~Rahaman, S.~Jaques and V.~I.~Paulsen, \emph{Eventually entanglement breaking maps}, J. Math. Phys. \textbf{59} (2018) 062201; \href{https://arxiv.org/abs/1801.05542}{arXiv:1801.05542}.

\bibitem{CMW2019} M.~Christandl, A.~M\"uller-Hermes and M.~M.~Wolf, \emph{When do composed maps become entanglement breaking?}, Ann. Henri Poincar\'e \textbf{20} (2019) 2295; \href{https://arxiv.org/abs/1807.01266}{arXiv:1807.01266}.

\bibitem{Park2026} S.-J.~Park, \emph{Every PPT channel has finite entanglement-breaking index}, \href{https://arxiv.org/abs/2608.13551}{arXiv:2608.13551} (2026).

\bibitem{Gittsovich2008} O.~Gittsovich, O.~G\"uhne, P.~Hyllus and J.~Eisert, \emph{Unifying several separability conditions using the covariance matrix criterion}, Phys. Rev. A \textbf{78} (2008) 052319; \href{https://arxiv.org/abs/0803.0757}{arXiv:0803.0757}.

\bibitem{Zhang2008} C.-J.~Zhang, Y.-S.~Zhang, S.~Zhang and G.-C.~Guo, \emph{Entanglement detection beyond the computable cross-norm or realignment criterion}, Phys. Rev. A \textbf{77} (2008) 060301(R); \href{https://arxiv.org/abs/0709.3766}{arXiv:0709.3766}.

\bibitem{MZ2026} H.~Ma and W.~Zhang, \emph{Optimal coefficient and nonconstant equality channels in every dimension} (2026), unpublished note, with a Lean~4 formalisation of the constant-channel argument in the repository \href{https://github.com/the-omega-institute/trureturing}{github.com/the-omega-institute/trureturing}, revision \texttt{c4a0b34}, file \texttt{D5/S3/Quantum/QuantumChannels/}\allowbreak\texttt{EntanglementBreakingNormCoefficient.lean}.

\end{thebibliography}
\end{document}